\pdfoutput=1
\documentclass{amsart}
\usepackage{hyperref}

\usepackage{tikz}
\usepackage{tikz-3dplot}
\usepackage{braids}

\usetikzlibrary{decorations.markings, perspective}

\tikzset{
  baseline=(current bounding box.center),
  x=28pt, y=28pt, font=\small,
  >=stealth,
  vertex/.style={inner sep=0pt},
  lvertex/.style={vertex, draw, fill= white, opacity=1, minimum size=8pt, font=\scriptsize},
  gnode/.style={lvertex, shape=circle},
  fnode/.style={lvertex, shape=rectangle},
  ->-/.style={decoration={
      markings, mark=at position #1 with
      {\arrow{>}}},postaction={decorate}},
  -<-/.style={decoration={
      markings, mark=at position #1 with
      {\arrow{<}}},postaction={decorate}},
}
 
\newcommand{\Z}{\mathbb{Z}}
\newcommand{\R}{\mathbb{R}}
\newcommand{\C}{\mathbb{C}}
\newcommand{\V}{\mathbb{V}}

\newcommand{\CC}{\mathcal{C}}
\newcommand{\CD}{\mathcal{D}}
\newcommand{\CW}{\mathcal{W}}

\DeclareMathOperator{\Tr}{Tr}
\DeclareMathOperator{\End}{End}
\let\Re\relax
\DeclareMathOperator{\Re}{Re}

\newcommand{\Tb}{\overline{T}}
\newcommand{\Cb}{\overline{C}}

\newcommand{\tet}[7]{T_{#1}\biggl[
  \arraycolsep=1pt
  \begin{array}{ccc}
    #2 & #3 & #4 \\
    #5 & #6 & #7
  \end{array}
\biggr]}

\newcommand{\tetb}[7]{\overline{T}_{#1}\biggl[
  \arraycolsep=1pt
  \begin{array}{ccc}
    #2 & #3 & #4 \\
    #5 & #6 & #7
  \end{array}
\biggr]}

\newcommand{\ttet}[7]{{}^tT_{#1}\biggl[
  \arraycolsep=1pt
  \begin{array}{ccc}
    #2 & #3 & #4 \\
    #5 & #6 & #7
  \end{array}
\biggr]}

\newcommand{\Bol}[9]{W_{#1}\biggl[
  \arraycolsep=1pt
  \begin{array}{cccc}
    #2 & #3 & #4 & #5 \\
    #6 & #7 & #8 & #9
  \end{array}
\biggr]}

\newcommand{\sBol}[9]{\CW_{#1}\biggl[
  \arraycolsep=1pt
  \begin{array}{cccc}
    #2 & #3 & #4 & #5 \\
    #6 & #7 & #8 & #9
  \end{array}
\biggr]}

\newtheorem{proposition}{Proposition}

\theoremstyle{remark}
\newtheorem{remark}[proposition]{Remark}

\title{State integral models and the tetrahedron equation}

\author{Junya Yagi}
\address{Yau Mathematical Sciences Center, Tsinghua University, China}

\date{\today}

\begin{document}
\begin{abstract}
  It is shown that for a class of state integral models on shaped
  pseudo $3$-manifolds, including the edge formulation of
  Teichm\"uller TQFT, the Boltzmann weight assigned to a tetrahedron
  solves the tetrahedron equation. The dihedral angles of the
  tetrahedron play the role of spectral parameters.
\end{abstract}

\maketitle

\section{Introduction}

The tetrahedron equation \cite{MR611994b, Zamolodchikov:1981kf} is a
highly overdetermined system of nonlinear equations that underlies
integrability in three dimensions.  It is substantially more complex
than its two-dimensional counterpart, the Yang--Baxter equation, and
only a limited number of solutions are known.

Among the known solutions, of particular interest are those involving
quantum dilogarithm functions.  Such solutions have been constructed
in various frameworks, such as quantized algebras of functions
\cite{MR1278735}, discrete differential geometry
\cite{Bazhanov:2005as, Bazhanov:2008rd}, and cluster algebras
\cite{Sun:2022mpy, Inoue:2023vtx, Inoue:2023rer, Inoue:2024swb}.  The
fact that different approaches give rise to the same solutions
indicates the presence of deep connections among the relevant
mathematical structures.  Furthermore, some of these solutions arise
from supersymmetric gauge theories \cite{Sun:2022mpy} and branes in
M-theory \cite{Yagi:2022tot}.

Quantum dilogarithms also appear in a class of three-dimensional state
integral models, including Teichm\"uller TQFT \cite{MR3227503, AK2}.
These are statistical mechanics models defined on shaped pseudo
$3$-manifolds, which consist of tetrahedra, each assigned the shape of
an ideal hyperbolic tetrahedron.  In these state integral models, the
Boltzmann weight is built from a quantum dilogarithm function and
satisfies the shaped pentagon identity ensuring the invariance of the
partition function under $2$--$3$ moves between triangulations.

Since quantum dilogarithms are crucial ingredients in both cases, it
is natural to ask whether there is a way to realize solutions of the
tetrahedron equation, and hence three-dimensional integrable lattice
models, by state integral models.

In this paper we show that for a state integral model with state
variables residing on the edges of tetrahedra, the tetrahedral weight
(that is, the Boltzmann weight for a single tetrahedron) solves the
tetrahedron equation in the Interaction-Round-a-Cube (IRC) form,
provided that the shaped pentagon identity is satisfied both by the
tetrahedral weight and by its transpose, in which the state variables
on each pair of opposite edges are exchanged.  The dihedral angles of
the tetrahedron play the role of spectral parameters.

The use of state integral models to construct solutions of the
tetrahedron equation was advocated in \cite{Shim:2026uxv}, inspired in
part by the work of Maillet~\cite{Maillet:1993av}.  The tetrahedron
equation in the IRC form can be interpreted geometrically as an
equivalence between two ways of decomposing a rhombic dodecahedron
into four parallelepipeds.  See Fig.~\ref{fig:TE-IRC}.  The idea of
\cite{Shim:2026uxv} is to take this interpretation seriously and
construct solutions as partition functions of state integral models on
certain shaped triangulations of parallelepipeds.  The resulting
solutions, however, satisfy only a weaker form of the tetrahedron
equation.  Moreover, their dependence on spectral parameters can be
eliminated by a gauge transformation on dihedral angles.

\begin{figure}
  \begin{equation*}
  \begin{tikzpicture}[3d view={0}{60}, scale=0.75]
    \node[gnode] (d) at (0,0,0) {$d$};

    \node[gnode] (a1) at (-1,-1,-1) {$a_1$};
    \node[gnode] (a2) at (-1,1,1) {$a_2$};
    \node[gnode] (a3) at (1,-1,1) {$a_3$};
    \node[gnode] (a4) at (1,1,-1) {$a_4$};
    
    \node[gnode] (b1) at (1,1,1) {$b_1$};
    \node[gnode] (b2) at (1,-1,-1) {$b_2$};
    \node[gnode] (b3) at (-1,1,-1) {$b_3$};
    \node[gnode] (b4) at (-1,-1,1) {$b_4$};

    \node[gnode] (c1) at (2,0,0) {$c_1$};
    \node[gnode] (c2) at (0,0,-2) {$c_2$};
    \node[gnode] (c3) at (0,2,0) {$c_3$};
    \node[gnode] (c4) at (0,0,2) {$c_4$};
    \node[gnode] (c5) at (-2,0,0) {$c_5$};
    \node[gnode] (c6) at (0,-2,0) {$c_6$};
    
    \draw[black!25] (a1) -- (c2);
    \draw[black!25] (a4) -- (c1);
    \draw[black!25] (a4) -- (c2);
    \draw[black!25] (a4) -- (c3);
    \draw[black!25] (b2) -- (c2);
    \draw[black!25] (b3) -- (c2);

    \draw[black!50] (b1) -- (d);
    \draw[black!50] (b2) -- (d);
    \draw[black!50] (b3) -- (c3);
    \draw[black!50] (b3) -- (c5);
    \draw[black!50] (b3) -- (d);
    \draw[black!50] (b4) -- (d);

    \draw (a1) -- (c5);
    \draw (a1) -- (c6);
    \draw (a2) -- (c4);
    \draw (a2) -- (c3);
    \draw (a2) -- (c5);
    \draw (a3) -- (c6);
    \draw (a3) -- (c1);
    \draw (a3) -- (c4);
    \draw (b1) -- (c1);
    \draw (b1) -- (c3);
    \draw (b1) -- (c4);
    \draw (b2) -- (c1);
    \draw (b2) -- (c6);
    \draw (b4) -- (c4);
    \draw (b4) -- (c5);
    \draw (b4) -- (c6);
  \end{tikzpicture}
  \ = \ 
  \begin{tikzpicture}[3d view={0}{60}, scale=0.75]
    \node[gnode] (d) at (0,0,0) {$d$};

    \node[gnode] (a1) at (-1,-1,-1) {$a_1$};
    \node[gnode] (a2) at (-1,1,1) {$a_2$};
    \node[gnode] (a3) at (1,-1,1) {$a_3$};
    \node[gnode] (a4) at (1,1,-1) {$a_4$};
    
    \node[gnode] (b1) at (1,1,1) {$b_1$};
    \node[gnode] (b2) at (1,-1,-1) {$b_2$};
    \node[gnode] (b3) at (-1,1,-1) {$b_3$};
    \node[gnode] (b4) at (-1,-1,1) {$b_4$};

    \node[gnode] (c1) at (2,0,0) {$c_1$};
    \node[gnode] (c2) at (0,0,-2) {$c_2$};
    \node[gnode] (c3) at (0,2,0) {$c_3$};
    \node[gnode] (c4) at (0,0,2) {$c_4$};
    \node[gnode] (c5) at (-2,0,0) {$c_5$};
    \node[gnode] (c6) at (0,-2,0) {$c_6$};
    
    \draw[black!25] (a1) -- (c2);
    \draw[black!25] (a4) -- (c2);
    \draw[black!25] (b2) -- (c2);
    \draw[black!25] (b3) -- (c2);
    \draw[black!25] (b3) -- (c3);
    \draw[black!25] (b3) -- (c5);

    \draw[black!50] (a1) -- (d);
    \draw[black!50] (a2) -- (d);
    \draw[black!50] (a3) -- (d);
    \draw[black!50] (a4) -- (c1);
    \draw[black!50] (a4) -- (c3);
    \draw[black!50] (a4) -- (d);

    \draw (a1) -- (c5);
    \draw (a1) -- (c6);
    \draw (a2) -- (c3);
    \draw (a2) -- (c4);
    \draw (a2) -- (c5);
    \draw (a3) -- (c1);
    \draw (a3) -- (c4);
    \draw (a3) -- (c6);
    \draw (b1) -- (c1);
    \draw (b1) -- (c3);
    \draw (b1) -- (c4);
    \draw (b2) -- (c1);
    \draw (b2) -- (c6);
    \draw (b4) -- (c4);
    \draw (b4) -- (c5);
    \draw (b4) -- (c6);
  \end{tikzpicture}
  \end{equation*}
    
  \caption{The tetrahedron equation in the IRC form.}
  \label{fig:TE-IRC}
\end{figure}
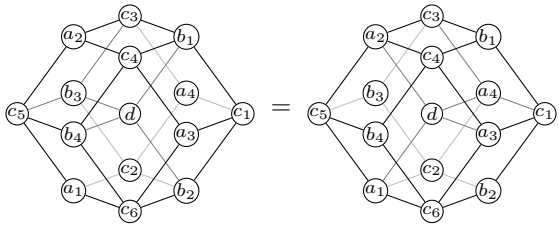

The construction presented in this work is based on a different idea,
one motivated by the cluster algebra approach of \cite{Sun:2022mpy}.
In this construction, the state variables on the edges of tetrahedra
in state integral models are mapped to state variables at the vertices
of cubes in IRC models; the cubic lattice does not correspond to any
shaped pseudo $3$-manifold.  For solutions produced with this
construction, the tetrahedron equation is satisfied in the strict
form, and the spectral parameters are genuine physical parameters.

In a recent work \cite{BKMS}, it is claimed that solutions of the
tetrahedron equation in the IRC form as well as in the vertex form can
be constructed from quantum dilogarithms satisfying an inversion
relation and the pentagon identity.  It is plausible that their
construction is related to the one discussed here.

The rest of the paper is organized as follows. We begin in
section~\ref{sec:IRC} by reviewing IRC models and their integrability,
which is implied by the tetrahedron equation.  In
section~\ref{sec:SIM}, we introduce state integral models on shaped
pseudo $3$-manifolds. We present the main result in
section~\ref{sec:TE-TW}.  Section~\ref{sec:conclusions}
contains concluding remarks.

\subsection*{Acknowledgments}

The author would like to thank Myungbo Shim, Xiaoyue Sun and Wang Hao
for discussions.  This work is supported by NSFC Grant Number
12375064.

\section{Integrable IRC models}
\label{sec:IRC}

\subsection{IRC models}

Fix positive integers $L$, $M$, $N$, and let
\begin{equation}
  T^3_{L,M,N} := (\R/L\Z) \times (\R/M\Z) \times (\R/N\Z)
\end{equation}
be the $3$-torus of size $L \times M \times N$.  Inside $T^3_{L,M,N}$,
consider the $L+M+N$ planes
\begin{alignat}{2}
  X_l &:= \{(l, y, z) \mid y \in \R/M\Z, \ z \in \R/N\Z\} \,,
  &\quad
  l &\in \Z/L\Z \,,
  \\
  Y_m &:= \{(x, m, z) \mid x \in \R/L\Z, \ z \in \R/N\Z\} \,,
  &\quad
  m &\in \Z/M\Z \,,
  \\
  Z_n &:= \{(x, y, n) \mid x \in \R/L\Z, \ y \in \R/M\Z\} \,,
  &\quad
  n &\in \Z/N\Z \,.
\end{alignat}
The lines of intersection of these planes form an
$L \times M \times N$ periodic cubic lattice.  Connecting the centers
of the cubes by lines orthogonal to faces, we obtain the dual lattice,
which is again an $L \times M \times N$ cubic lattice and has vertices
at the points $(l+\frac12, m+\frac12, n+\frac12)$, $l \in \Z/L\Z$,
$m \in \Z/M\Z$, $n \in \Z/N\Z$.

An \emph{IRC model} on the dual cubic lattice is a statistical
mechanics model whose state variables are placed at the vertices and
interact locally around the cubes.  The interaction energy is a
function of the state variables, which take values in a set $S$; for
the purpose of the present paper we can assume that $S$ is a subset of
$\C$.  The interaction energy may depend on parameters assigned to the
planes, called the \emph{spectral parameters}.  Typically, spectral
parameters are complex numbers or points on a complex algebraic curve.
We will use $C$ to denote the set of their possible values.

Let $a_{x,y,z} \in S$ be the state variable at
$(x,y,z) \in T^3_{L,M,N}$, and let $s_l$, $t_m$, $u_n \in C$ be
the spectral parameters assigned to the planes $X_l$, $Y_m$,
$Z_n$, respectively.  At the eight corners of the cube centered at
$(l,m,n)$ where $X_l$, $Y_m$, $Z_n$ intersect, the eight
state variables $a_{l\pm\frac12,m\pm\frac12,n\pm\frac12}$ are present
and interact.  We denote the Boltzmann weight for this local
interaction by
\begin{equation}
  \label{eq:W}
  \Bol{s_l,t_m,u_n}{a_{l-\frac12,m-\frac12,n-\frac12}}{a_{l+\frac12,m-\frac12,n-\frac12}}{a_{l-\frac12,m+\frac12,n-\frac12}}{a_{l-\frac12,m-\frac12,n+\frac12}}{a_{l+\frac12,m+\frac12,n+\frac12}}{a_{l-\frac12,m+\frac12,n+\frac12}}{a_{l+\frac12,m-\frac12,n+\frac12}}{a_{l+\frac12,m+\frac12,n-\frac12}} \,;
\end{equation}
we will also abbreviate it as
\begin{equation}
  W^{(l,m,n)} \,.
\end{equation}
In traditional statistical mechanics, the Boltzmann weight is
interpreted as a probability measure and hence is a nonnegative real
number.  In this work we will consider models with complex Boltzmann
weights.  Thus, the local Boltzmann weight~\eqref{eq:W} defines a
function
\begin{equation}
  W\colon C^3 \times S^8 \to \C \,.
\end{equation}
See Fig.~\ref{fig:W} for an illustration of the local Boltzmann
weight.

\begin{figure}
  \begin{equation*}
    \Bol{s,t,u}{a}{b}{c}{d}{e}{f}{g}{h}
    = \quad
\begin{tikzpicture}[3d view={145}{20}, scale=1.2, font=\scriptsize]
      \begin{scope}[xshift=40pt]
        \draw[->] (0,0,0) -- (0.3,0,0) node[left=-2pt, yshift=-1pt] {$x$};
        \draw[->] (0,0,0) -- (0,0.3,0) node[right=-2pt, yshift=-1pt] {$y$};
        \draw[->] (0,0,0) -- (0,0,0.3) node[above=-2pt] {$z$};
      \end{scope}
      
      \node[gnode] (a) at (0,0,0) {$a$};
      \node[gnode] (b) at (1,0,0) {$b$};
      \node[gnode] (c) at (0,1,0) {$c$};
      \node[gnode] (d) at (0,0,1) {$d$};

      \draw[black!50] (a) -- (b);
      \draw[black!50] (a) -- (c);
      \draw[black!50] (a) -- (d);

      \draw[green!20!black, fill=green, opacity=0.2]
      (-0.1,0.5,-0.1) -- (-0.1,0.5,1.1)
      -- (1.1,0.5,1.1) -- (1.1,0.5,-0.1) node[below left=-2pt, opacity=1] {$t$} -- cycle;
      
      \draw[red!20!black, fill=red, opacity=0.2]
      (0.5,-0.1,-0.1) -- (0.5,1.1,-0.1)
      -- (0.5,1.1,1.1) -- (0.5,-0.1,1.1) node[above left=-2pt, opacity=1] {$s$} -- cycle;
      
      \draw[blue!20!black, fill=blue, opacity=0.2]
      (-0.1,-0.1,0.5) -- (1.1,-0.1,0.5)
      -- (1.1,1.1,0.5) -- (-0.1,1.1,0.5) node[right=-2pt, opacity=1] {$u$} -- cycle;

      \node[gnode] (f) at (0,1,1) {$f$};
      \node[gnode] (g) at (1,0,1) {$g$};
      \node[gnode] (h) at (1,1,0) {$h$};
      \node[gnode] (e) at (1,1,1) {$e$};

      \draw (f) -- (e);
      \draw (g) -- (e);
      \draw (h) -- (e);
      \draw (b) -- (g);
      \draw (b) -- (h);
      \draw (c) -- (h);
      \draw (c) -- (f);
      \draw (d) -- (f);
      \draw (d) -- (g);

      \draw[red!50!green, opacity=0.2] (0.5,0.5,-0.1) -- (0.5,0.5,1.1);
      \draw[red!50!blue, opacity=0.2] (0.5,-0.1,0.5) -- (0.5,1.1,0.5);
      \draw[green!50!blue, opacity=0.2] (-0.1,0.5,0.5) -- (1.1,0.5,0.5);
    \end{tikzpicture}
    \end{equation*}
    
    \caption{The local Boltzmann weight for an IRC model.}
    \label{fig:W}
\end{figure}

The Boltzmann weight for a configuration of state variables on the
entire lattice is given by the product of all local Boltzmann weights
from all cubes in the lattice.  The \emph{partition function} $Z$ of
the model is the integral of the Boltzmann weight over all
configurations of state variables:
\begin{equation}
  Z
  :=
  \int \prod_{l,m,n}  W^{(l,m,n)} d\mu(a_{l-\frac12,m-\frac12,n-\frac12})
  \,.
\end{equation}
Here the integration over each state variable is performed with
respect to some measure $\mu$ on $S \subset \C$.

\subsection{Integrability and the tetrahedron equation}

In order to state the integrability condition for the model, we
reformulate the model using a quantum mechanical language.

Consider the layer of $L \times M$ cubes intersected by $Z_n$.  The
Boltzmann weight for this layer is determined by $2LM$ state
variables, $LM$ each on the top and the bottom of the layer.  Let us
write this Boltzmann weight as
\begin{equation}
  \tau(u_n; s_1, \dots, s_L, t_1, \dots, t_M)_{a_{\frac12,\frac12,n-\frac12} \dotso a_{L-\frac12,M-\frac12,n-\frac12}}^{a_{\frac12,\frac12,n+\frac12} \dotso a_{L-\frac12,M-\frac12,n+\frac12}}
  :=
  \prod_{l,m} W^{(l,m,n)}
  \,.
\end{equation}
Viewing this quantity as a matrix element of a linear map acting on
the vector space $\V(S^{LM})$ of an appropriate class of functions on
$S^{LM}$ (where summation over indices is replaced by integration with
respect to $\mu$), we obtain a map
\begin{equation}
  \tau\colon C \times C^L \times C^M
  \to \End\bigl(\V(S^{LM})\bigr) \,,
\end{equation}
which we call the \emph{layer transfer matrix}.

Using the layer transfer matrix, we can calculate the partition
function as
\begin{equation}
  \label{eq:Z}
  Z = \Tr_{\V(S^{LM})}\bigl(\tau(u_N) \tau(u_{N-1}) \dotsm \tau(u_1)\bigr) \,.
\end{equation}
(Here and below we suppress from the notation the dependence on the
spectral parameters $s_1, \dots, s_L$, $t_1, \dots, t_M$.)  The factor
$\tau(u_n)$ ``transfers'' a state on the bottom of the $n$th layer to
one on the top.  In this sense, the layer transfer matrix can be
thought of as a discrete analog of the time evolution operator in
quantum mechanics.  After the initial state on the bottom of the first
layer is transferred to the top of the $N$th layer, the final state is
projected onto the initial state by the trace since the time direction
is periodic.

Suppose that $\tau(u)$ is analytic in $u$, and $\tau(u)$ and
$\tau(u')$ commute for any values $u$, $u'$ of the spectral
parameters:
\begin{equation}
  \label{eq:TT}
  [\tau(u), \tau(u')] = 0 \,.
\end{equation}
The commutativity implies that the coefficients of the expansion of
$\tau(u)$ in powers of $u$ are mutually commuting operators on
$\V(S^{LM})$.  In this case, we have a family of commuting
``conserved charges'' and the model is said to be \emph{integrable}.

A fundamental result in the area of three-dimensional lattice models
is that the commutativity \eqref{eq:TT} holds if the Boltzmann weight
satisfies the tetrahedron equation and a certain operator constructed
from the Boltzmann weight is invertible.  This is the
three-dimensional analogue of the train track argument in
two-dimensional integrable lattice models, where the Yang--Baxter
equation allows two row transfer matrices to be interchanged.  See,
e.g., \cite{Inoue:2025xzj} for a graphical proof of this result.

The tetrahedron equation is an equivalence between two configurations
of four intersecting planes $H_1$, $H_2$, $H_3$, $H_4$,
where each configuration cuts out a tetrahedron in three-dimensional
space.  To each intersection of three planes a factor of the local
Boltzmann weight is assigned.  For IRC models, the tetrahedron equation reads \cite{MR683099,
  MR696804}
\begin{equation}
  \label{eq:TE}
  \begin{aligned}
    \int_S d\mu(d)
    &
    \Bol{r_1, r_2, r_3}{a_4}{c_2}{c_3}{c_1}{d}{b_1}{b_2}{b_3}
    \Bol{r_1, r_2, r_4}{c_1}{b_2}{b_1}{a_3}{b_4}{c_4}{c_6}{d}
    \\
    \times
    &
    \Bol{r_1, r_3, r_4}{b_1}{d}{c_3}{c_4}{c_5}{a_2}{b_4}{b_3}
    \Bol{r_2, r_3, r_4}{d}{b_2}{b_3}{b_4}{a_1}{c_5}{c_6}{c_2}
    \\
    =
    \int_S d\mu(d)
    &
    \Bol{r_2, r_3, r_4}{b_1}{c_1}{c_3}{c_4}{d}{a_2}{a_3}{a_4}
    \Bol{r_1, r_3, r_4}{c_1}{b_2}{a_4}{a_3}{a_1}{d}{c_6}{c_2}
    \\
    \times
    &
    \Bol{r_1, r_2, r_4}{a_4}{c_2}{c_3}{d}{c_5}{a_2}{a_1}{b_3}
    \Bol{r_1, r_2, r_3}{d}{a_1}{a_2}{a_3}{b_4}{c_4}{c_6}{c_5}
    \,,
  \end{aligned}
\end{equation}
where $r_i$ is the spectral parameter of $H_i$.  This is an equality
to be satisfied for all
$(a_1, \dots, a_4, b_1, \dots, b_4, c_1,\dots, c_6) \in S^{14}$; these
state variables reside at the vertices of the rhombic dodecahedra in
Fig.~\ref{fig:TE-IRC}.

Alternatively, the tetrahedron equation may be understood as an
equivalence between the two sequences of braid moves shown in
Fig.~\ref{fig:TE-wires}, which convert one wiring diagram with four
wires $H_1$, $H_2$, $H_3$, $H_4$ that represents a reduced word for
the longest element of the symmetric group on four letters into
another wiring diagram representing another reduced word for the same
element.  Each braid move corresponds to a factor of the local
Boltzmann weight.  The four factors on the left-hand side
of~\eqref{eq:TE} correspond, from left to right, to the four arrows in
the top row of Fig.~\ref{fig:TE-wires}; the four factors on the
right-hand side correspond similarly to the arrows in the bottom row.

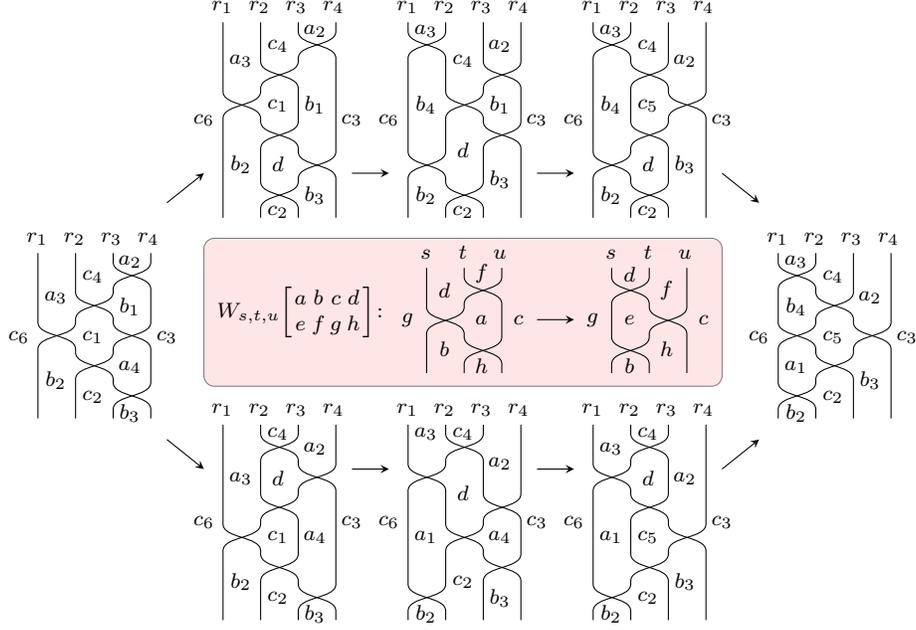
\begin{figure}
  \centering
  \begin{tikzpicture}[xscale=0.5, yscale=0.4, font=\footnotesize]
    \begin{scope}[shift={(8,-1.5)}]
      \draw[rounded corners, draw=black!50, fill=red!10] (-2.5,1) rectangle (11.5,-4);

      \node at (0.2,-1.5) {$\Bol{s,t,u}{a}{b}{c}{d}{e}{f}{g}{h}\colon$};

      \begin{scope}[shift={(2.5,0)}]
        \braid[gap=0] a_2 a_1 a_2;
        
        \node at (1,0) [above] {$s$};
        \node at (2,0) [above] {$t$};
        \node at (3,0) [above] {$u$};
        
        \node at (0.5,-1.75) {$g$};
        \node at (2.5,-1.75) {$a$};
        \node at (1.5,-0.75) {$d$};
        \node at (1.5,-2.75) {$b$};
        \node at (2.5,-3.3) {$h$};
        \node at (2.5,-0.2) {$f$};
        \node at (3.5,-1.75) {$c$};
      \end{scope}
      \draw[->, shift={(6.5,-1.75)}] (0,0) -- (1,0);
      \begin{scope}[shift={(7.5,0)}]
        \braid[gap=0] a_1 a_2 a_1;
        
        \node at (1,0) [above] {$s$};
        \node at (2,0) [above] {$t$};
        \node at (3,0) [above] {$u$};
        
        \node at (0.5,-1.75) {$g$};
        \node at (1.5,-1.75) {$e$};
        \node at (1.5,-0.2) {$d$};
        \node at (1.5,-3.3) {$b$};
        \node at (2.5,-2.75) {$h$};
        \node at (2.5,-0.75) {$f$};
        \node at (3.5,-1.75) {$c$};
      \end{scope}
    \end{scope}

    \begin{scope}[shift={(0,-1)}]
      \braid[gap=0] a_3 a_2 a_3-a_1 a_2 a_3;
      
      \node at (1,0.5) {$r_1$};
      \node at (2,0.5) {$r_2$};
      \node at (3,0.5) {$r_3$};
      \node at (4,0.5) {$r_4$};
      
      \node at (0.5,-2.85) {$c_6$};
      
      \node at (1.5,-1.5) {$a_3$};
      \node at (1.5,-4.3) {$b_2$};
      
      \node at (2.5,-0.8) {$c_4$};
      \node at (2.5,-2.85) {$c_1$};
      \node at (2.5,-4.85) {$c_2$};
      
      \node at (3.5,-0.3) {$a_2$};
      \node at (3.5,-1.8) {$b_1$};
      \node at (3.5,-3.85) {$a_4$};
      \node at (3.5,-5.4) {$b_3$};
      
      \node at (4.5,-2.85) {$c_3$};
    \end{scope}
    
    \begin{scope}[shift={(5,6.8)}] 
      \braid[gap=0] a_3 a_2 a_1 a_2 a_3 a_2;
      
      \node at (1,0.5) {$r_1$};
      \node at (2,0.5) {$r_2$};
      \node at (3,0.5) {$r_3$};
      \node at (4,0.5) {$r_4$};
      
      \node at (0.5,-3.3) {$c_6$};
      
      \node at (1.5,-1.3) {$a_3$};
      \node at (1.5,-4.8) {$b_2$};
      
      \node at (2.5,-0.8) {$c_4$};
      \node at (2.5,-2.8) {$c_1$};
      \node at (2.5,-4.8) {$d$};
      \node at (2.5,-6.4) {$c_2$};
      
      \node at (3.5,-0.3) {$a_2$};
      \node at (3.5,-2.8) {$b_1$};
      \node at (3.5,-5.85) {$b_3$};
      
      \node at (4.5,-3.3) {$c_3$};
    \end{scope}
    
    \begin{scope}[shift={(10,6.8)}] 
      \braid[gap=0] a_1 a_3 a_2 a_3 a_1 a_2;
      
      \node at (1,0.5) {$r_1$};
      \node at (2,0.5) {$r_2$};
      \node at (3,0.5) {$r_3$};
      \node at (4,0.5) {$r_4$};
        
      \node at (0.5,-3.3) {$c_6$};
      
      \node at (1.5,-0.3) {$a_3$};
      \node at (1.5,-2.8) {$b_4$};
      \node at (1.5,-5.85) {$b_2$};
      
      \node at (2.5,-1.3) {$c_4$};
      \node at (2.5,-6.4) {$c_2$};
      \node at (2.5,-4.3) {$d$};
      
      \node at (3.5,-0.8) {$a_2$};
      \node at (3.5,-2.8) {$b_1$};
      \node at (3.5,-5.3) {$b_3$};
      
      \node at (4.5,-3.3) {$c_3$};
    \end{scope}
    
    \begin{scope}[shift={(15,6.8)}] 
      \braid[gap=0] a_1 a_2 a_3 a_2 a_1 a_2;
      
      \node at (1,0.5) {$r_1$};
      \node at (2,0.5) {$r_2$};
      \node at (3,0.5) {$r_3$};
      \node at (4,0.5) {$r_4$};
      
      \node at (0.5,-3.3) {$c_6$};
      
      \node at (1.5,-0.3) {$a_3$};
      \node at (1.5,-2.8) {$b_4$};
      \node at (1.5,-5.85) {$b_2$};
      
      \node at (2.5,-0.8) {$c_4$};
      \node at (2.5,-2.8) {$c_5$};
      \node at (2.5,-6.4) {$c_2$};
      \node at (2.5,-4.8) {$d$};
      
      \node at (3.5,-1.3) {$a_2$};
      \node at (3.5,-4.8) {$b_3$};
      
      \node at (4.5,-3.3) {$c_3$};
    \end{scope}
    
    \begin{scope}[shift={(5,-6.8)}] 
      \braid[gap=0] a_2 a_3 a_2 a_1 a_2 a_3;
      
      \node at (1,0.5) {$r_1$};
      \node at (2,0.5) {$r_2$};
      \node at (3,0.5) {$r_3$};
      \node at (4,0.5) {$r_4$};
      
      \node at (0.5,-3.3) {$c_6$};
      
      \node at (1.5,-1.8) {$a_3$};
      \node at (1.5,-5.3) {$b_2$};
      
      \node at (2.5,-0.3) {$c_4$};
      \node at (2.5,-1.8) {$d$};
      \node at (2.5,-3.8) {$c_1$};
      \node at (2.5,-5.85) {$c_2$};
      
      \node at (3.5,-0.8) {$a_2$};
      \node at (3.5,-3.8) {$a_4$};
      \node at (3.5,-6.4) {$b_3$};
      
      \node at (4.5,-3.3) {$c_3$};
    \end{scope}
    
    \begin{scope}[shift={(10,-6.8)}] 
      \braid[gap=0] a_2 a_1 a_3 a_2 a_3 a_1;
      
      \node at (1,0.5) {$r_1$};
      \node at (2,0.5) {$r_2$};
      \node at (3,0.5) {$r_3$};
      \node at (4,0.5) {$r_4$};
      
      \node at (0.5,-3.3) {$c_6$};
      
      \node at (1.5,-0.3) {$a_3$};
      \node at (1.5,-3.8) {$a_1$};
      \node at (1.5,-6.4) {$b_2$};
      
      \node at (2.5,-0.3) {$c_4$};
      \node at (2.5,-2.3) {$d$};
      \node at (2.5,-5.3) {$c_2$};
      
      \node at (3.5,-1.3) {$a_2$};
      \node at (3.5,-3.8) {$a_4$};
      \node at (3.5,-5.85) {$b_3$};
      
      \node at (4.5,-3.3) {$c_3$};
    \end{scope}
    
    \begin{scope}[shift={(15,-6.8)}] 
      \braid[gap=0] a_2 a_1 a_2 a_3 a_2 a_1;
      
      \node at (1,0.5) {$r_1$};
      \node at (2,0.5) {$r_2$};
      \node at (3,0.5) {$r_3$};
      \node at (4,0.5) {$r_4$};
      
      \node at (0.5,-3.3) {$c_6$};
      
      \node at (1.5,-0.8) {$a_3$};
      \node at (1.5,-3.8) {$a_1$};
      \node at (1.5,-6.4) {$b_2$};
      
      \node at (2.5,-0.3) {$c_4$};
      \node at (2.5,-1.8) {$d$};
      \node at (2.5,-3.8) {$c_5$};
      \node at (2.5,-5.85) {$c_2$};
      
      \node at (3.5,-1.8) {$a_2$};
      \node at (3.5,-5.3) {$b_3$};
      
      \node at (4.5,-3.3) {$c_3$};
    \end{scope}
    
    \begin{scope}[shift={(20,-1)}] 
      \braid[gap=0] a_1 a_2 a_3-a_1 a_2 a_1;
      
      \node at (1,0.5) {$r_1$};
      \node at (2,0.5) {$r_2$};
      \node at (3,0.5) {$r_3$};
      \node at (4,0.5) {$r_4$};
      
      \node at (0.5,-2.85) {$c_6$};
      
      \node at (1.5,-0.3) {$a_3$};
      \node at (1.5,-3.85) {$a_1$};
      \node at (1.5,-1.8) {$b_4$};
      \node at (1.5,-5.4) {$b_2$};
      
      \node at (2.5,-0.8) {$c_4$};
      \node at (2.5,-2.85) {$c_5$};
      \node at (2.5,-4.85) {$c_2$};
      
      \node at (3.5,-1.5) {$a_2$};
      \node at (3.5,-4.3) {$b_3$};
      
      \node at (4.5,-2.85) {$c_3$};
    \end{scope}
    
    \draw[->, shift={(4.5,4)}] (0,-3.3) -- (1,-2.3);
    \draw[->, shift={(4.5,-4)}] (0,-3.3) -- (1,-4.3);
    
    \draw[->, shift={(9.5,5)}] (0,-3.3) -- (1,-3.3);
    \draw[->, shift={(9.5,-5)}] (0,-3.3) -- (1,-3.3);
    
    \draw[->, shift={(14.5,5)}] (0,-3.3) -- (1,-3.3);
    \draw[->, shift={(14.5,-5)}] (0,-3.3) -- (1,-3.3);
    
    \draw[->, shift={(19.5,4)}] (0,-2.3) -- (1,-3.3);
    \draw[->, shift={(19.5,-4)}] (0,-4.3) -- (1,-3.3);
  \end{tikzpicture}
  \caption{The wiring diagram presentation of the local Boltzmann
    weight \eqref{eq:W} and the tetrahedron equation \eqref{eq:TE}.}
  \label{fig:TE-wires}
\end{figure}

\subsection{Six-parameter tetrahedron equation}

A stronger form of the tetrahedron equation depends on six parameters
$\rho_{12}$, $\rho_{13}$, $\rho_{14}$, $\rho_{23}$, $\rho_{24}$,
$\rho_{34}$:
\begin{equation}
  \label{eq:TE6}
  \begin{aligned}
    \int_S d\mu(d)
    &
    \sBol{\rho_{12}, \rho_{13}, \rho_{23}}{a_4}{c_2}{c_3}{c_1}{d}{b_1}{b_2}{b_3}
    \sBol{\rho_{12}, \rho_{14}, \rho_{24}}{c_1}{b_2}{b_1}{a_3}{b_4}{c_4}{c_6}{d}
    \\
    \times
    &
    \sBol{\rho_{13}, \rho_{14}, \rho_{34}}{b_1}{d}{c_3}{c_4}{c_5}{a_2}{b_4}{b_3}
    \sBol{\rho_{23}, \rho_{24}, \rho_{34}}{d}{b_2}{b_3}{b_4}{a_1}{c_5}{c_6}{c_2}
    \\
    =
    \int_S d\mu(d)
    &
    \sBol{\rho_{23}, \rho_{24}, \rho_{34}}{b_1}{c_1}{c_3}{c_4}{d}{a_2}{a_3}{a_4}
    \sBol{\rho_{13}, \rho_{14}, \rho_{34}}{c_1}{b_2}{a_4}{a_3}{a_1}{d}{c_6}{c_2}
    \\
    \times
    &
    \sBol{\rho_{12}, \rho_{14}, \rho_{24}}{a_4}{c_2}{c_3}{d}{c_5}{a_2}{a_1}{b_3}
    \sBol{\rho_{12}, \rho_{13}, \rho_{23}}{d}{a_1}{a_2}{a_3}{b_4}{c_4}{c_6}{c_5}
    \,.
  \end{aligned}
\end{equation}
The parameter $\rho_{ij}$ may be regarded as assigned to the
intersection $H_i \cap H_j$.  This is the form in which the
tetrahedron equation originally appeared in \cite{MR611994b,
  Zamolodchikov:1981kf, MR696804},%
\footnote{The tetrahedron equation \eqref{eq:TE6} coincides with the
  relation given in Eqs.~(2.2), (3.14), (3.15) of \cite{MR696804}
  under the identification $\rho_{12} = \theta_1$,
  $\rho_{13} = \pi - \theta_3$, $\rho_{14} = \theta_2$,
  $\rho_{23} = \theta_4$, $\rho_{24} = \pi - \theta_6$,
  $\rho_{34} = \theta_5$ and
  $\sBol{s,t,u}{a}{b}{c}{d}{e}{f}{g}{h} = W(u, s, \pi-t;
  a|bdc|fhg|e)$.}
though in Zamolodchikov's solution the six parameters satisfy one
relation.

For simplicity, assume that all six parameters are independent. Let
$\CC$ denote the set of their values.  Then, given a solution
\begin{equation}
  \CW\colon \CC^3 \times S^8 \to \C
\end{equation}
of the six-parameter tetrahedron equation \eqref{eq:TE6}, setting
\begin{equation}
  W_{r_i,r_j,r_k} = \CW_{\rho(r_i,r_j), \rho(r_i,r_k), \rho(r_j,r_k)}
\end{equation}
for any function $\rho\colon C^2 \to \CC$ yields a solution of
the tetrahedron equation \eqref{eq:TE} with four spectral parameters.

\section{State integral models on shaped pseudo $3$-manifolds}
\label{sec:SIM}

\subsection{Shaped pseudo $3$-manifolds}

Imagine that we have a finite collection of tetrahedra, each with
totally ordered vertices, and we glue them along faces in such a way
that each edge can be consistently oriented from the smaller to larger
vertices in all tetrahedra containing it.  The resulting space is
known as an oriented triangulated pseudo $3$-manifold.  The collection
of tetrahedra, now embedded in this space, is called its
triangulation.

The triangulation of an oriented triangulated pseudo $3$-manifold is
said to be \emph{shaped} if it is endowed with a \emph{shape
  structure}, an assignment of the dihedral angles of an ideal
hyperbolic tetrahedron to each tetrahedron in the triangulation: if
the vertices of the tetrahedron are $0$, $1$, $2$, $3$ and
$\alpha_v \in \R_{>0}$ ($v = 1$, $2$ or $3$) is the dihedral angle
around edge $0v$, then $\alpha_v$ is also the dihedral angle around
the edge opposite to $0v$ and
\begin{equation}
  \alpha_1 + \alpha_2 + \alpha_3 = \pi \,.
\end{equation}
A \emph{shaped pseudo $3$-manifold} is an oriented triangulated pseudo
$3$-manifold with a shape structure.

\subsection{Tetrahedral weights}

Consider a statistical mechanics model whose state variables take
values in a subset $S$ of $\C$, endowed with a measure $\mu$, and are
located on the edges of a shaped pseudo $3$-manifold.  Take a single
tetrahedron in the triangulation and label its vertices $0$, $1$, $2$,
$3$ according to their ordering.  Let
$\alpha = (\alpha_1, \alpha_2, \alpha_3)$ be the triplet of dihedral
angles of this tetrahedron (defined as in the previous paragraph), and
let $x_v$ and $x'_v$ denote the state variables on edge $0v$ and the
opposite edge, respectively.

To this tetrahedron the model assigns a Boltzmann weight, which may
depend on the orientation of the tetrahedron.  The orientation is
determined by the relative positions of the vertices: viewed from
vertex $0$, the vertices $1$, $2$, $3$ are arranged either clockwise
or counterclockwise, as illustrated in Fig.~\ref{fig:tetrahedron}.  In
the former case the tetrahedron is positive, and in the latter
negative.

\begin{figure}
  \small
  \begin{equation*} 
   \tet{\alpha}{x_1}{x_2}{x_3}{x'_1}{x'_2}{x'_3}
    =
  \begin{tikzpicture}[3d view={-40}{20}, scale=1.2, rotate=0, font=\scriptsize]
    \coordinate (0) at (0,0,1.2);
    \coordinate (3) at (1,0,0);
    \coordinate (1) at (-1/2,-{sqrt(3)/2},0); 
    \coordinate (2) at (-1/2,{sqrt(3)/2},0);
    
    \draw [->-=0.8] (2) -- node[pos=-0.05] {2} node[xshift=-6pt, sloped, above=-2pt] {$\alpha_1$, $x'_1$} (3) node[pos=1.05] {3};

    \fill[black!10, opacity=0.6] (0) -- (3) -- (1) -- (2) -- cycle;

    \draw [->-=0.8] (0) node[above] {0}
    -- node[xshift=-4pt, sloped, above=-2pt] {$\alpha_1$, $x_1$} (1) node[pos=1.05] {1};
    \draw [->-=0.8] (0)
    -- node[xshift=0pt, sloped, above=-2pt] {$\alpha_2$, $x_2$} (2);
    \draw [->-=0.8] (0)
    -- node[xshift=-2pt, sloped, above=-2pt] {$\alpha_3$, $x_3$} (3);
    \draw [->-=0.8] (1) --  node[xshift=0pt, sloped, below=-2pt] {$\alpha_3$, $x'_3$} (2);
    \draw [->-=0.8] (1)
    -- node[xshift=0pt, sloped, below=-2pt] {$\alpha_2$, $x'_2$} (3);
  \end{tikzpicture}
  \qquad
  \tetb{\alpha}{x_1}{x_2}{x_3}{x'_1}{x'_2}{x'_3}
  =
  \begin{tikzpicture}[3d view={-40}{20}, scale=1.2, rotate=0, font=\scriptsize]
    \coordinate (0) at (0,0,1.2);
    \coordinate (2) at (1,0,0);
    \coordinate (1) at (-1/2,-{sqrt(3)/2},0); 
    \coordinate (3) at (-1/2,{sqrt(3)/2},0);
    
    \draw [->-=0.8] (2) -- node[pos=-0.05] {2} node[xshift=-6pt, sloped, above=-2pt] {$\alpha_1$, $x'_1$} (3) node[pos=1.05] {3};

    \fill[black!10, opacity=0.6] (0) -- (3) -- (1) -- (2) -- cycle;

    \draw [->-=0.8] (0) node[above] {0}
    -- node[xshift=-4pt, sloped, above=-2pt] {$\alpha_1$, $x_1$} (1) node[pos=1.05] {1};
    \draw [->-=0.8] (0)
    -- node[xshift=-2pt, sloped, above=-2pt] {$\alpha_2$, $x_2$} (2);
    \draw [->-=0.8] (0)
    -- node[xshift=0pt, sloped, above=-2pt] {$\alpha_3$, $x_3$} (3);
    \draw [->-=0.8] (1) --  node[xshift=0pt, sloped, below=-2pt] {$\alpha_3$, $x'_3$} (2);
    \draw [->-=0.8] (1)
    -- node[xshift=-4pt, sloped, below=-2pt] {$\alpha_2$, $x'_2$} (3);
  \end{tikzpicture}
  \end{equation*}

  \caption{The tetrahedral weights for a positive tetrahedron (left)
    and a negative tetrahedron (right).}
  \label{fig:tetrahedron}
\end{figure}
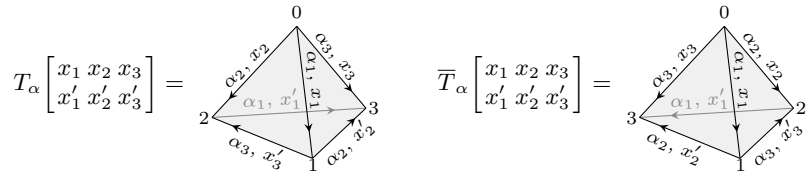

If the tetrahedron is positive, we denote its Boltzmann weight by
\begin{equation}
  \label{eq:T}
  \tet{\alpha}{x_1}{x_2}{x_3}{x'_1}{x'_2}{x'_3} \,.
\end{equation}
Varying the dihedral angles and the state variables, we obtain a
function
\begin{equation}
  T\colon A \times S^6 \to \C \,,
\end{equation}
where
\begin{equation}
  A := \{(\alpha_1,\alpha_2,\alpha_3) \in \R_{>0}^3
  \mid \alpha_1 + \alpha_2 + \alpha_3 = \pi\}
\end{equation}
is the domain of dihedral angles.  Similarly, the Boltzmann weight
assigned to a negative tetrahedron defines a function
\begin{equation}
  \Tb\colon A \times S^6 \to \C \,,
  \quad
  (\alpha; x_1, x_2, x_3, x'_1, x'_2, x'_3)
  \mapsto
  \tetb{\alpha}{x_1}{x_2}{x_3}{x'_1}{x'_2}{x'_3} \,.
\end{equation}
We call these functions the \emph{tetrahedral weights} of the model.
The two tetrahedral weights $T$ and $\Tb$ are typically related by
complex conjugation.

\subsection{Shaped pentagon identity}

Now suppose that a $3$-manifold is homeomorphic to a shaped pseudo
$3$-manifold such that each internal edge is balanced, i.e., the sum
of dihedral angles around it is $2\pi$.  Then, the same $3$-manifold
is homeomorphic to different shaped pseudo $3$-manifolds with balanced
internal edges, related by shaped versions of Pachner moves between
triangulations.  If the partition function of the model is the same on
these different shaped triangulations, we have a \emph{state integral
model} which is expected to capture a topological invariant of the
underlying $3$-manifold.

For this to happen, the partition function must in particular be
invariant under shaped $2$--$3$ moves.  A shaped $2$--$3$ move
combines two tetrahedra into a bipyramid and then decomposes it into
three tetrahedra (and vice versa), as depicted in Fig.~\ref{fig:2-3}.

The invariance under shaped $2$--$3$ moves requires the tetrahedral
weights to satisfy the \emph{shaped pentagon identity}
\begin{align}
  \label{eq:P}
  \tet{\alpha^{(1)}}{x_{02}}{x_{03}}{x_{04}}{x_{34}}{x_{24}}{x_{23}}
  &
  \tet{\alpha^{(3)}}{x_{01}}{x_{02}}{x_{04}}{x_{24}}{x_{14}}{x_{12}}
  \\ \nonumber
  =
  C_{\alpha^{(0)}, \alpha^{(2)}, \alpha^{(4)}}
  \int_S d\mu(x_{13})
  &
    \tet{\alpha^{(0)}}{x_{12}}{x_{13}}{x_{14}}{x_{34}}{x_{24}}{x_{23}}
  \tet{\alpha^{(2)}}{x_{01}}{x_{03}}{x_{04}}{x_{34}}{x_{14}}{x_{13}}
  \tet{\alpha^{(4)}}{x_{01}}{x_{02}}{x_{03}}{x_{23}}{x_{13}}{x_{12}}
    \,,
  \\
  \tetb{\alpha^{(1)}}{x_{02}}{x_{03}}{x_{04}}{x_{34}}{x_{24}}{x_{23}}
  &
  \tetb{\alpha^{(3)}}{x_{01}}{x_{02}}{x_{04}}{x_{24}}{x_{14}}{x_{12}}
  \\ \nonumber
  =
  \Cb_{\alpha^{(0)}, \alpha^{(2)}, \alpha^{(4)}}
  \int_S d\mu(x_{13})
  &
  \tetb{\alpha^{(0)}}{x_{12}}{x_{13}}{x_{14}}{x_{34}}{x_{24}}{x_{23}}
  \tetb{\alpha^{(2)}}{x_{01}}{x_{03}}{x_{04}}{x_{34}}{x_{14}}{x_{13}}
  \tetb{\alpha^{(4)}}{x_{01}}{x_{02}}{x_{03}}{x_{23}}{x_{13}}{x_{12}}
\end{align}
whenever the dihedral angles satisfy the compatibility condition
\begin{equation}
  \label{eq:P-angles}
  \begin{alignedat}{2}
    \alpha^{(1)}_1 &= \alpha^{(0)}_1 + \alpha^{(2)}_1 \,,
    &\qquad
    \alpha^{(1)}_3 &= \alpha^{(0)}_3 + \alpha^{(4)}_1 \,,
    \\
    \alpha^{(3)}_1 &= \alpha^{(2)}_1 + \alpha^{(4)}_1 \,,
    &\qquad
    \alpha^{(3)}_3 &= \alpha^{(0)}_1 + \alpha^{(4)}_3 \,,
    \\
    & & \alpha^{(2)}_3 &= \alpha^{(1)}_3 + \alpha^{(3)}_3 \,.
  \end{alignedat}
\end{equation}
Here we have allowed coefficients
$C_{\alpha^{(0)}, \alpha^{(2)}, \alpha^{(4)}}$,
$\Cb_{\alpha^{(0)}, \alpha^{(2)}, \alpha^{(4)}}$ depending on dihedral
angles to appear on the right-hand sides.  If the partition function
is strictly invariant under shaped $2$--$3$ moves, then
$C_{\alpha^{(0)}, \alpha^{(2)}, \alpha^{(4)}} = \Cb_{\alpha^{(0)},
  \alpha^{(2)}, \alpha^{(4)}} = 1$.

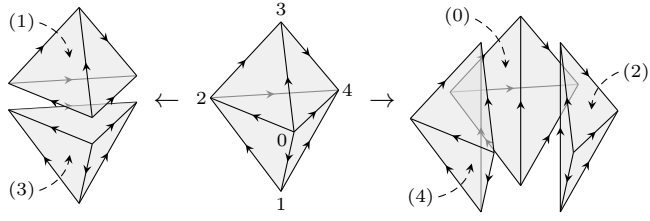
\begin{figure}
\begin{equation*}
  \begin{tikzpicture}[3d view={-40}{20}, scale=1, rotate=0, font=\scriptsize]

  \begin{scope}[yshift=-5pt]
    \coordinate (0) at (-1/2,-{sqrt(3)/2},0); 
    \coordinate (1) at (0,0,-1.2);
    \coordinate (2) at (-1/2,{sqrt(3)/2},0);
    \coordinate (3) at (0,0,1.2);
    \coordinate (4) at (1,0,0);
    
    \draw [->-=0.5] (2) -- (4);
    \fill[black!10, opacity=0.6] (1) -- (2) -- (4) -- cycle;
    \draw [->-=0.5] (0) -- (1);
    \draw [->-=0.5] (0) -- (2);
    \draw [->-=0.5] (0) -- (4);
    \draw [->-=0.5] (1) -- (2);
    \draw [->-=0.5] (1) -- (4);

    \draw[densely dashed, ->] (-2,-1.5,0) node[left=-2pt] {$(3)$} to[bend right] (-1,-1,0);
  \end{scope}
  \begin{scope}[yshift=5pt]
    \coordinate (0) at (-1/2,-{sqrt(3)/2},0); 
    \coordinate (1) at (0,0,-1.2);
    \coordinate (2) at (-1/2,{sqrt(3)/2},0);
    \coordinate (3) at (0,0,1.2);
    \coordinate (4) at (1,0,0);
    
    \draw [->-=0.5] (2) -- (4);
    \fill[black!10, opacity=0.6] (0) -- (2) -- (3) -- (4) -- cycle;
    \draw [->-=0.5] (0) -- (2);
    \draw [->-=0.5] (0) -- (3);
    \draw [->-=0.5] (0) -- (4);
    \draw [->-=0.5] (2) -- (3);
    \draw [->-=0.5] (3) -- (4);

    \draw[densely dashed, ->] (-2.3,-1.5,1.7) node[above=-2pt] {$(1)$} to[bend right] (-1,-0.8,0.8);
  \end{scope}
\end{tikzpicture}
\ \leftarrow
\begin{tikzpicture}[3d view={-40}{20}, scale=1, rotate=0, font=\scriptsize]
  \coordinate (0) at (-1/2,-{sqrt(3)/2},0); 
  \coordinate (1) at (0,0,-1.2);
  \coordinate (2) at (-1/2,{sqrt(3)/2},0);
  \coordinate (3) at (0,0,1.2);
  \coordinate (4) at (1,0,0);
  
  \draw [->-=0.5] (2) -- node[pos=-0.07] {$2$} (4) node[pos=1.07] {$4$};
  \fill[black!10, opacity=0.6] (1) -- (2) -- (3) -- (4) -- cycle;
  \draw [->-=0.5] (0) -- (1) node[below=-1pt] {$1$};
  \draw [->-=0.5] (0) -- (2);
  \draw [->-=0.5] (0) node[xshift=-4.5pt, yshift=-3pt] {$0$} -- (3) node[above=-1pt] {$3$};
  \draw [->-=0.5] (0) -- (4);
  \draw [->-=0.5] (1) -- (2);
  \draw [->-=0.5] (1) -- (4);
  \draw [->-=0.5] (2) -- (3);
  \draw [->-=0.5] (3) -- (4);
\end{tikzpicture}
\rightarrow\!
\begin{tikzpicture}[3d view={-40}{20}, scale=1, rotate=0, font=\scriptsize]
  \begin{scope}[yshift=10pt]
    \draw[densely dashed, ->] (-2.3,-1.8,2.2) node[left=-2pt, yshift=1pt] {$(0)$} to[bend left] (-1.1,-1,1);

    \coordinate (0) at (-1/2,-{sqrt(3)/2},0); 
    \coordinate (1) at (0,0,-1.2);
    \coordinate (2) at (-1/2,{sqrt(3)/2},0);
    \coordinate (3) at (0,0,1.2);
    \coordinate (4) at (1,0,0);

    \draw [->-=0.5] (2) -- (4);
    \fill[black!10, opacity=0.6] (1) -- (2) -- (3) -- (4) -- cycle;
    \draw [->-=0.5] (1) -- (2);
    \draw [->-=0.5] (1) -- (4);
    \draw [->-=0.5] (1) -- (3);
    \draw [->-=0.5] (2) -- (3);
    \draw [->-=0.5] (3) -- (4);
  \end{scope}

  \begin{scope}[xshift=15pt]
    \coordinate (0) at (-1/2,-{sqrt(3)/2},0); 
    \coordinate (1) at (0,0,-1.2);
    \coordinate (2) at (-1/2,{sqrt(3)/2},0);
    \coordinate (3) at (0,0,1.2);
    \coordinate (4) at (1,0,0);

    \fill[black!10, opacity=0.6] (1) -- (3) -- (4) -- cycle;

    \draw [->-=0.5] (0) -- (1);
    \draw [->-=0.5] (0) -- (3);
    \draw [->-=0.5] (0) -- (4);
    \draw [->-=0.5] (1) -- (4);
    \draw [->-=0.5] (1) -- (3);
    \draw [->-=0.5] (3) -- (4);

    \draw[densely dashed, ->] (-0.5,-1.8,1.4) node[right=-2pt] {$(2)$} to[bend right] (-0.3,-1,0.6);
\end{scope}

  \begin{scope}[xshift=-15pt]
    \coordinate (0) at (-1/2,-{sqrt(3)/2},0); 
    \coordinate (1) at (0,0,-1.2);
    \coordinate (2) at (-1/2,{sqrt(3)/2},0);
    \coordinate (3) at (0,0,1.2);
    \coordinate (4) at (1,0,0);

    \draw [->-=0.5] (1) -- (3);
    \fill[black!10, opacity=0.6] (1) -- (2) -- (3) -- (0) -- cycle;
    \draw [->-=0.5] (0) -- (1);
    \draw [->-=0.5] (0) -- (2);
    \draw [->-=0.5] (0) -- (3);
    \draw [->-=0.5] (1) -- (2);
    \draw [->-=0.5] (2) -- (3);

    \draw[densely dashed, ->] (-2,-1.5,-0.1) node[left=-2pt] {$(4)$} to[bend right] (-1,-1,0);
  \end{scope}
\end{tikzpicture}
\end{equation*}

\caption{A shaped $2$--$3$ move for positive tetrahedra. The vertices
  of the bipyramid in the middle are labeled by the number of incoming
  edges, including an invisible edge from vertex~1 to vertex~3.  In
  the shaped pentagon identity~\eqref{eq:P}, the tetrahedron on the
  left or right that does not contain vertex~$v$ is labeled $(v)$.}
  \label{fig:2-3}
\end{figure}

\subsection{Tetrahedral symmetry}

If the tetrahedral weights are required to depend on the orientation
of the tetrahedron but not on the vertex ordering itself, then they
must be invariant under the $\Z_2$ rotation symmetry
\begin{equation}
  \label{eq:Z2}
  \tet{\alpha}{x_1}{x_2}{x_3}{x'_1}{x'_2}{x'_3}
  =
 \tet{\alpha}{x_1}{x'_2}{x'_3}{x'_1}{x_2}{x_3}
\end{equation}
and the $\Z_3$ rotation symmetry
\begin{equation}
  \label{eq:Z3}
  \tet{(\alpha_1,\alpha_2,\alpha_3)}{x_1}{x_2}{x_3}{x'_1}{x'_2}{x'_3}
  =
  \tet{(\alpha_2,\alpha_3,\alpha_1)}{x_2}{x_3}{x_1}{x'_2}{x'_3}{x'_1}
  \,.
\end{equation}
These rotations together generate the chiral tetrahedral symmetry
$A_4$, which acts on the tetrahedron by even permutations of the four
vertices.  Odd permutations in the full tetrahedral group $S_4$ flip
the orientation of the tetrahedron.

Although it is natural to expect the tetrahedral weights to possess
the chiral tetrahedral symmetry, this is not always the case.
Teichm\"uller TQFT is an example in which the chiral tetrahedral
symmetry is absent at the level of tetrahedral weights, while the
partition functions on closed shaped pseudo $3$-manifolds still enjoy
this symmetry.

\subsection{Transpose of the tetrahedral weight}

Regarding $T$ as an operator on $\V(S^3)$, we define its transpose
${}^tT$ by
\begin{equation}
  \ttet{\alpha}{x_1}{x_2}{x_3}{x'_1}{x'_2}{x'_3}
  :=
  \tet{\alpha}{x'_1}{x'_2}{x'_3}{x_1}{x_2}{x_3} \,.
\end{equation}
It turns out that, for the construction of solutions of the
tetrahedron equation, the transpose tetrahedral weight ${}^tT$ must
also satisfy the shaped pentagon identity~\eqref{eq:P}.

\subsection{Examples}

\subsubsection{Meromorphic 3D index}

This is a state integral model describing Chern--Simons theory with
gauge group $\mathrm{SL}(2,\C)$ and at level $k = 0$.  It was
constructed by Garoufalidis and Kashaev \cite{MR3882990}, building on
pioneering work by Dimofte, Gaiotto and Gukov \cite{Dimofte:2011ju,
  Dimofte:2011py} and subsequent developments \cite{MR3522084,
  MR3416111}.

In this model, the state variables take values on the unit circle:
\begin{equation}
  S = \{z \in \C \mid |z| = 1\} \,,
  \qquad
  d\mu(z) = \frac{dz}{2\pi iz} \,.
\end{equation}
The tetrahedral weight depends on a complex parameter $q$ with
$|q| < 1$ and is given by
\begin{equation}
  \tet{\alpha}{x_1}{x_2}{x_3}{x'_1}{x'_2}{x'_3}
  =
  \frac{(q;q)_\infty^2}{(q^2;q^2)_\infty}
  \prod_{v=1}^3
  G_q\biggl((-q)^{\alpha_v/\pi} \frac{x_{v-1} x'_{v-1}}{x_{v+1} x'_{v+1}}\biggr)
  \,,
\end{equation}
where $(z;q)_\infty := \prod_{i=0}^\infty (1 - q^iz)$ is the
$q$-Pochhammer symbol,
\begin{equation}
  G_q(z) := \frac{(-qz^{-1};q)_\infty}{(z;q)_\infty} \,,
\end{equation}
and the indices on the state variables are understood modulo $3$.  The
chiral tetrahedral symmetry~\eqref{eq:Z2}, \eqref{eq:Z3} is manifest,
and the pentagon identity \eqref{eq:P} is satisfied with
$C_{\alpha^{(0)}, \alpha^{(2)}, \alpha^{(4)}} = 1$.

\subsubsection{Kashaev--Luo--Vartanov model}

Another example, which has a structure analogous to the 3D index
model, is the state integral model constructed by Kashaev, Luo and
Vartanov \cite{MR3486430}.

The model depends on a parameter $b \in \C$ with $\Re b > 0$.  Let
\begin{equation}
  \Phi_b(z)
  :=
  \exp\biggl(
  \int_{\R + i0}
  \frac{e^{-2ixz}}{4\sinh(xb) \sinh(xb^{-1})}
  \frac{dx}{x}
  \biggr)
\end{equation}
be Faddeev's quantum dilogarithm and
\begin{equation}
  \Psi_b(x) := e^{-\pi ix^2/2} \frac{\Phi_b(x)}{\Phi_b(0)}
\end{equation}
be the normalized version satisfying $\Psi_b(x) \Psi_b(-x) = 1$.

For this model, the state variables are real-valued,
\begin{equation}
  S = \R \,,
  \qquad
  d\mu(x) = dx \,.
\end{equation}
As in the case of the 3D index model, the tetrahedral weight satisfies
the pentagon identity \eqref{eq:P} with
$C_{\alpha^{(0)}, \alpha^{(2)}, \alpha^{(4)}} = 1$ and has the chiral
tetrahedral symmetry \cite{Kashaev-RIMS}:
\begin{equation}
  \tet{\alpha}{x_1}{x_2}{x_3}{x'_1}{x'_2}{x'_3}
  =
  \prod_{v=1}^3
  \Psi_b\biggl(
  x_{v+1} + x'_{v+1} - x_{v-1} - x'_{v-1}
  + i(b + b^{-1})\Bigl(\frac12 - \frac{\alpha_v}{\pi}\Bigr)
  \biggr)
  \,.
\end{equation}

\subsubsection{Teichm\"uller TQFT}

This state integral model was proposed by Andersen and Kashaev
\cite{MR3227503, AK2} building on earlier works by Hikami
\cite{Hikami:2001en, MR2330673}, Dimofte \cite{MR3250765}, Dimofte,
Gukov, Lenells and Zagier \cite{MR2551896}, Dijkgraaf, Fuji and Manabe
\cite{Dijkgraaf:2010ur}, among others.  It captures
$\mathrm{SL}(2,\C)$ Chern--Simons theory at level $k = 1$.  The
partition function of the Kashaev--Luo--Vartanov model is conjectured
to be equal to two times the absolute value squared of the partition
function of Teichm\"uller TQFT.

In the edge formulation of Teichm\"uller TQFT \cite{AK2}, the state
variables take values in the unit interval:
\begin{equation}
  S = [0,1] \,,
  \qquad
  d\mu(x) = dx \,.
\end{equation}
The tetrahedral weight is given by
\begin{equation}
  \tet{\alpha}{x_1}{x_2}{x_3}{x'_1}{x'_2}{x'_3}
  =
  g_{\alpha_1/2\pi, \alpha_3/2\pi}
  (x_2 + x'_2 - x_3 - x'_3, x_2 + x'_2 - x_1 - x'_1) \,,
\end{equation}
where $g_{a,c}$ is a certain function constructed from $\Phi_b$, and
satisfies the pentagon identity \eqref{eq:P} with
\begin{equation}
  C_{\alpha^{(0)}, \alpha^{(2)}, \alpha^{(4)}}
  = e^{-\pi i/12} e^{-\pi i(\alpha^{(0)}_3 + \alpha^{(2)}_1 + \alpha^{(4)}_3)/6}
  \,.
\end{equation}
It has the $\Z_2$ rotation symmetry~\eqref{eq:Z2}, but is not quite
invariant under the $\Z_3$ rotations \cite[Eq.~(33)]{AK2}:
\begin{equation}
  \tet{(\alpha_1,\alpha_2,\alpha_3)}{x_1}{x_2}{x_3}{x'_1}{x'_2}{x'_3}
  =
  e^{\pi i/12}
  e^{\pi i(x_2 + x'_2 - x_3 - x'_3)/2}
  \tet{(\alpha_2,\alpha_3,\alpha_1)}{x_2}{x_3}{x_1+\frac12}{x'_2}{x'_3}{x'_1}
  \,.
\end{equation}

\section{Solutions of the tetrahedron equation from tetrahedral
  weights}
\label{sec:TE-TW}

Let $T\colon A \times S^6 \to \C$ be the tetrahedral weight of a state
integral model.  For some set $\CC$, choose a function
$\alpha\colon \CC^3 \to \R^3$ and let $\CD \subset \CC^6$ be a domain
such that for every
$(\rho_{12}, \rho_{13}, \rho_{14}, \rho_{23}, \rho_{24}, \rho_{34})
\in \CD$,
\begin{equation}
  \label{eq:alphas}
  \begin{alignedat}{3}
    \alpha^{(1)} &= \alpha(\rho_{12}, \rho_{13}, \rho_{23}) \,,
    & \qquad &
    \alpha^{(2)} &= \alpha(\rho_{12}, \rho_{14}, \rho_{24}) \,,
    \\
    \alpha^{(3)} &= \alpha(\rho_{13}, \rho_{14}, \rho_{34}) \,,
    & \qquad &
    \alpha^{(4)} &= \alpha(\rho_{23}, \rho_{24}, \rho_{34})
  \end{alignedat}
\end{equation}
belong to $A$ and fulfil the compatibility
condition~\eqref{eq:P-angles} together with some $\alpha^{(0)} \in A$.
The chiral tetrahedral symmetry of $T$ is not assumed.

\begin{proposition}
  \label{prop:TE}

  Suppose that the shaped pentagon identity~\eqref{eq:P} holds for
  both $T$ and its transpose ${}^tT$ with the same coefficient
  $C_{\alpha^{(0)}, \alpha^{(2)}, \alpha^{(4)}}$.  Then the Boltzmann
  weight
  \begin{equation}
    \label{eq:Wrhorhorho}
    \sBol{\rho_{ij}, \rho_{ik}, \rho_{jk}}{a}{b}{c}{d}{e}{f}{g}{h}
    =
    \tet{\alpha(\rho_{ij}, \rho_{ik}, \rho_{jk})}{f}{a}{h}{b}{e}{d}
  \end{equation}
  solves the six-parameter tetrahedron equation \eqref{eq:TE6} on
  $\CD \times S^{14}$, provided that the order of integration can be
  interchanged in the passage from \eqref{eq:integral-1} to
  \eqref{eq:integral-2} below.
\end{proposition}

\begin{proof}
  Using the shaped pentagon identity \eqref{eq:P} for $T$ and ${}^tT$,
  the left-hand side of~\eqref{eq:TE6} can be turned into the
  right-hand side as follows:
  \begin{align}
    &
    \int_S d\mu(d)
    \tet{\alpha^{(1)}}{b_1}{a_4}{b_3}{c_2}{d}{c_1}
    \tet{\alpha^{(2)}}{c_4}{c_1}{d}{b_2}{b_4}{a_3}
    \tet{\alpha^{(3)}}{a_2}{b_1}{b_3}{d}{c_5}{c_4}
    \tet{\alpha^{(4)}}{c_5}{d}{c_2}{b_2}{a_1}{b_4}
    \\
    \label{eq:integral-1}
    &
    =
    \int_S d\mu(d)
    \tet{\alpha^{(2)}}{c_4}{c_1}{d}{b_2}{b_4}{a_3}
    \tet{\alpha^{(4)}}{c_5}{d}{c_2}{b_2}{a_1}{b_4}
    \\ \nonumber
    & \qquad \times
    C_{\alpha^{(0)}, \alpha^{(2)}, \alpha^{(4)}}
    \int_S d\mu(e)
    \tet{\alpha^{(0)}}{c_4}{e}{c_5}{c_2}{d}{c_1}
    \tet{\alpha^{(2)}}{a_2}{a_4}{b_3}{c_2}{c_5}{e}
    \tet{\alpha^{(4)}}{a_2}{b_1}{a_4}{c_1}{e}{c_4}
    \\
    &
    \label{eq:integral-2}
    =
    \int_S d\mu(e)
    \tet{\alpha^{(2)}}{a_2}{a_4}{b_3}{c_2}{c_5}{e}
    \tet{\alpha^{(4)}}{a_2}{b_1}{a_4}{c_1}{e}{c_4}
    \\ \nonumber
    & \qquad \times
    C_{\alpha^{(0)}, \alpha^{(2)}, \alpha^{(4)}}
    \int_S d\mu(d)
    \ttet{\alpha^{(0)}}{c_2}{d}{c_1}{c_4}{e}{c_5}
    \ttet{\alpha^{(2)}}{b_2}{b_4}{a_3}{c_4}{c_1}{d}
    \ttet{\alpha^{(4)}}{b_2}{a_1}{b_4}{c_5}{d}{c_2}
    \\
    &=
    \int_S d\mu(e)
    \tet{\alpha^{(2)}}{a_2}{a_4}{b_3}{c_2}{c_5}{e}
    \tet{\alpha^{(4)}}{a_2}{b_1}{a_4}{c_1}{e}{c_4}
    \ttet{\alpha^{(1)}}{a_1}{b_4}{a_3}{c_4}{e}{c_5}
    \ttet{\alpha^{(3)}}{b_2}{a_1}{a_3}{e}{c_1}{c_2}
    \\
    &=
    \int_S d\mu(d)
    \tet{\alpha^{(4)}}{a_2}{b_1}{a_4}{c_1}{d}{c_4}
    \tet{\alpha^{(3)}}{d}{c_1}{c_2}{b_2}{a_1}{a_3}
    \tet{\alpha^{(2)}}{a_2}{a_4}{b_3}{c_2}{c_5}{d}
    \tet{\alpha^{(1)}}{c_4}{d}{c_5}{a_1}{b_4}{a_3}
    \,.
  \end{align}
\end{proof}

All examples of state integral models in section~\ref{sec:SIM} satisfy
the assumptions of Proposition~\ref{prop:TE}.  The same shaped
pentagon identity holds for $T$ and ${}^tT$ because the tetrahedral
weights in these models are symmetric under transpose, as is manifest
from the explicit formulas for $T$.  The interchange of integrations
is justified as follows.  In the positive angle domain, the
tetrahedral weights are regular on the relevant integration domains.
For the meromorphic 3D index model and Teichm\"uller TQFT, the
compactness of the integration domains then implies that the integrand
is bounded.  For the Kashaev--Luo--Vartanov model, the standard
asymptotic estimates for $\Psi_b$ show exponential decay of the
integrand as $\|(d,e)\| \to \infty$.  In all cases the integrand is an
$L^1$-integrable function of the variables $(d,e)$, and Fubini's
theorem applies.

\begin{remark}
  The physical significance of the condition that both $T$ and ${}^tT$
  satisfy the shaped pentagon identity is unclear.  In fact, this
  condition seems rather restrictive.  For example, the F-symbols of a
  multiplicity-free fusion category define an angle-independent
  solution of the shaped pentagon identity by
  \begin{equation}
    \label{eq:tet-F}
    \tet{\alpha}{x_1}{x_2}{x_3}{x'_1}{x'_2}{x'_3}
    =
    [F^{x'_1 x'_3 x_1}_{x_3}]_{x'_2x_2}
    \,,
  \end{equation}
  where state variables are isomorphism classes of simple objects;
  however, its transpose fails to solve the shaped pentagon identity
  if the category has a nontrivial simple object.  Indeed, let $a$ be
  such an object and $b$ be a simple summand of $a \otimes a$ not
  isomorphic to $a$.  (Such a $b$ exists, for otherwise
  multiplicity-freeness would force $a \otimes a \cong a$, which would
  imply that $a$ is invertible and
  $a \cong a \otimes (a \otimes a^*) \cong (a \otimes a) \otimes a^*
  \cong \mathbf{1}$.)  With the choice
  \begin{equation}
      x_{01} = x_{34} = \mathbf{1} \,, \quad
      x_{02} = x_{04} = x_{24} = a \,, \quad
      x_{03} = x_{12} = x_{14} = x_{23} = b  \,,
  \end{equation}
  the left-hand side of the shaped pentagon identity for ${}^tT$ is
  nonzero, whereas the right-hand side vanishes.
\end{remark}

Let us give an example of the function $\alpha$.  Suppose that
$\CC \subset \R$ and $\alpha$ is affine.  In this case the
compatibility condition~\eqref{eq:P-angles} fixes the form of
$\alpha$, up to an overall rescaling of the spectral parameters:
\begin{equation}
  \label{eq:alpha}
  \alpha(\rho_{ij}, \rho_{ik}, \rho_{jk})
  = (\rho_{jk} - \rho_{ik}, \pi + \rho_{ij} - \rho_{jk} , \rho_{ik} - \rho_{ij})
  \,.
\end{equation}
The requirement that the dihedral angles of the tetrahedra appearing
in the tetrahedron equation, $\alpha^{(1)}$, $\alpha^{(2)}$,
$\alpha^{(3)}$, $\alpha^{(4)}$, all belong to $A$ is equivalent to the
inequalities
\begin{gather}
  \label{eq:rho-ineq-I}
  \rho_{12} < \rho_{13} < \min(\rho_{14}, \rho_{23}) \leq
  \max(\rho_{14}, \rho_{23}) < \rho_{24} < \rho_{34} \,,
  \\
  \label{eq:rho-ineq-II}
  \rho_{24} < \pi + \rho_{12} \,,
  \qquad
  \rho_{34} < \pi + \rho_{13} \,.
\end{gather}
The dihedral angles of the tetrahedron produced by the shaped $2$--$3$
move in the proof are
\begin{equation}
  \alpha^{(0)} = (\rho_{14} + \rho_{23} - \rho_{13} - \rho_{24}, \pi +
  \rho_{12} + \rho_{34} - \rho_{14} - \rho_{23}, \rho_{13} + \rho_{24} -
  \rho_{12} - \rho_{34}) \,,
\end{equation}
and requiring $\alpha^{(0)} \in A$ further imposes
\begin{equation}
  \label{eq:rho-ineq-III}
  \rho_{12} + \rho_{34} < \rho_{13} + \rho_{24} < \rho_{14} +
  \rho_{23}
  \,.
\end{equation}
The domain $\CD$ is therefore given by \eqref{eq:rho-ineq-I},
\eqref{eq:rho-ineq-II}, \eqref{eq:rho-ineq-III}.

If we wish to obtain a four-parameter solution of the tetrahedron
equation, we can set, for example,
\begin{equation}
  \rho_{ij} = r_i + r_j \,.
\end{equation}
With this specialization, the dependence of the Boltzmann weight on
the spectral parameters is only through their differences:
\begin{equation}
  \label{eq:Wrrr}
  \Bol{r_i,r_j,r_k}{a}{b}{c}{d}{e}{f}{g}{h}
  =
  \tet{(r_j - r_i, \pi + r_i - r_k, r_k - r_j)}{f}{a}{h}{b}{e}{d}
  \,.
\end{equation}
For this choice of $\rho_{ij}$, the inequalities
\eqref{eq:rho-ineq-I}, \eqref{eq:rho-ineq-II} combine into
\begin{equation}
  r_1 < r_2 < r_3 < r_4 < \pi + r_1 \,.
\end{equation}

The remaining condition \eqref{eq:rho-ineq-III}, however, collapses to
a trivial equality and $\alpha^{(0)}$ becomes $(0,\pi,0)$,
invalidating the assumption in Proposition~\ref{prop:TE} that
$\alpha^{(0)}$ belongs to $A$.  We can remedy this problem by first
perturbing $\rho_{12}$, $\rho_{13}$ to
\begin{equation}
  \rho_{12} = r_1 + r_2 - \epsilon \,,
  \qquad
  \rho_{13} = r_1 + r_3 - \delta
\end{equation}
with $\epsilon > \delta > 0$.  Then
$\alpha^{(0)} = (\delta, \pi - \epsilon, \epsilon - \delta)$ and
$\alpha^{(1)}$, $\alpha^{(2)}$, $\alpha^{(3)}$, $\alpha^{(4)}$ satisfy
the assumptions of Proposition~\ref{prop:TE} as long as
\begin{equation}
  r_1 < r_2 < r_3 < r_4 < \pi + r_1 - \epsilon \,.
\end{equation}
Since the four tetrahedra appearing in the tetrahedron equation remain
nondegenerate even at $\delta = \epsilon = 0$, the six-parameter
tetrahedron equation for the local Boltzmann
weight~\eqref{eq:Wrhorhorho} reduces to the tetrahedron equation
for~\eqref{eq:Wrrr}, under the assumptions that the tetrahedral weight
is continuous in the dihedral angles and that the limit $\epsilon$,
$\delta \to 0$ may be interchanged with the integration.  For the
examples in section~\ref{sec:SIM}, the former assumption follows from
regularity in the positive angle domain, and the latter is justified
by the dominated convergence theorem.

\begin{remark}
  In some state integral models, the partition function is invariant
  under shape gauge transformations%
  \footnote{On a single tetrahedron with dihedral angles
    $\alpha = (\alpha_1, \alpha_2, \alpha_3)$, the shape gauge
    transformation with parameter $\theta$ performed on the edge $0v$
    or on the opposite edge shifts the dihedral angle $\alpha_w$ by
    $\pm\sum_{x=1}^3 \varepsilon_{vwx} \theta$, where $\varepsilon$ is
    the completely antisymmetric tensor with $\varepsilon_{123} = 1$
    and the sign is correlated with the orientation of the
    tetrahedron.  On general shaped pseudo $3$-manifolds, shape gauge
    transformations act by shifting dihedral angles in each
    tetrahedron in the triangulation.}
  performed on internal edges.  In Teichm\"uller TQFT, for example,
  performing a shape gauge transformation on an edge is equivalent to
  shifting the state variable on that edge, which is an integration
  variable if the edge is internal.  A subgroup of the shape gauge
  symmetry descends to a gauge symmetry of the IRC models constructed
  from state integral models, but in general the spectral parameters
  cannot be removed by shape gauge transformations.  Indeed, the
  dihedral angles $\alpha^{(l,m,n)}$ of the tetrahedral weight
  corresponding to the cube~\eqref{eq:W} change under a shape gauge
  transformation by
  \begin{equation}
    \delta\alpha^{(l,m,n)}
    =
    \sum_{\sigma=\pm 1} \sum_{\eta = \pm\frac12}
    \sigma
    (
    \theta_{l-\eta,m-\eta,n+\sigma\eta},
    \theta_{l-\eta,m+\sigma\eta,n+\eta},
    \theta_{l+\sigma\eta,m+\eta,n+\eta}
    ) \,,
  \end{equation}
  where $\theta_{x,y,z}$ is the parameter assigned to the edge of the
  tetrahedron carrying the state variable $a_{x,y,z}$; therefore,
  $\sum_n \alpha^{(l,m,n)}_1$, $\sum_m \alpha^{(l,m,n)}_2$,
  $\sum_l \alpha^{(l,m,n)}_3$ are invariant on the periodic cubic
  lattice.  For the IRC model defined by~\eqref{eq:Wrrr},
  $\alpha^{(l,m,n)} = (t_m - s_l, \pi + s_l - u_n, u_n - t_m)$ and it
  follows that the difference of any pair of spectral parameters is
  gauge invariant.  In other words, shape gauge transformations only
  act by overall shifts of spectral parameters.
\end{remark}

\section{Conclusions}
\label{sec:conclusions}

In this work we constructed integrable IRC models from state integral
models.  Given a state integral model on shaped pseudo $3$-manifolds
with edge state variables, its tetrahedral weight provides a solution
of the tetrahedron equation if the tetrahedral weight and its
transpose both satisfy the shaped pentagon identity.  We have
discussed three examples of state integral models for which this
construction can be carried out.  To the best of the author's
knowledge, the resulting solutions have not appeared previously.

A distinctive feature of this construction is that the IRC model
associated with a state integral model is not defined by applying that
state integral model to a shaped triangulation of the cubic lattice.
In particular, in the IRC model state variables are placed at the
vertices of cubes, while in the original state integral model they are
placed on the edges of tetrahedra.  The local Boltzmann weight
constructed from the tetrahedral weight depends on only six of the
eight state variables at the corners of a cube.  This fact does not
make the model decoupled; any two state variables on the cubic lattice
can be joined by a chain of local Boltzmann weights since every state
variable appears in the local Boltzmann weights for some of the cubes
containing it.  In this sense, the construction produces intrinsically
three-dimensional integrable lattice models, albeit with generally
complex Boltzmann weights.

This work was motivated by the cluster algebra
approach~\cite{Sun:2022mpy}.  In that approach, solutions of the
tetrahedron equation arise from cluster transformations applied to
certain quivers.  The simplest solution corresponds to a single
mutation, which may be thought of geometrically as attaching a
tetrahedron to a triangulated surface.  Formulas in section~5
of~\cite{Sun:2022mpy} show that the Fock--Goncharov representation of
this solution agrees with the tetrahedral weight of the
Kashaev--Luo--Vartanov model in the flat limit.  The construction
developed here replaces this flat tetrahedral weight by a tetrahedral
weight with positive dihedral angles and extends the same idea to
other state integral models.

This observation naturally suggests the search for analogous
constructions corresponding to the other solutions of the tetrahedron
equation obtained in the cluster algebra approach.  Since several
important solutions involving quantum dilogarithms appear as
specializations of these solutions \cite{Inoue:2023vtx, Inoue:2023rer,
  Inoue:2024swb}, establishing such connections to state integral
models and understanding their physical implications would be
especially valuable.

\newcommand{\etalchar}[1]{$^{#1}$}
\providecommand{\bysame}{\leavevmode\hbox to3em{\hrulefill}\thinspace}
\providecommand{\MR}{\relax\ifhmode\unskip\space\fi MR }
\providecommand{\MRhref}[2]{%
  \href{http://www.ams.org/mathscinet-getitem?mr=#1}{#2}
}
\providecommand{\href}[2]{#2}

\end{document}